\documentclass[11pt, letterpaper]{article}
\usepackage[utf8]{inputenc}
\usepackage[margin=1in]{geometry}
\usepackage{amssymb}
\usepackage{amsthm}
\usepackage{amsmath}
\usepackage{mathtools}
\usepackage{bbm}
\usepackage{xcolor}
\usepackage{algorithm}
\usepackage{algpseudocode}
\usepackage[nodayofweek,level]{datetime}
\usepackage{graphicx}
\usepackage{hyperref}
\usepackage[capitalise]{cleveref}

\usepackage{amssymb}
\usepackage{amsthm}
\usepackage{amsmath}
\usepackage{mathtools}
\usepackage{bbm}
\usepackage{wasysym}

\newtheorem{theorem}{Theorem}
\newtheorem{lemma}{Lemma}

\newtheorem{corollary}{Corollary}

\theoremstyle{definition}
\newtheorem{definition}{Definition}

\theoremstyle{remark}
\newtheorem{remark}{Remark}

\DeclarePairedDelimiter{\ceil}{\lceil}{\rceil}
\DeclarePairedDelimiter{\floor}{\lfloor}{\rfloor}

\DeclarePairedDelimiter{\norm}{\|}{\|}

\newcommand{\ind}[1]{\mathbbm{1}\left[{#1}\right]}

\newcommand{\pr}{\mathbb{P}}
\newcommand{\E}{\mathbb{E}}

\newcommand{\R}{\mathbb{R}}

\newcommand{\N}{\mathbb{N}}

\title{\Large Multi-dimensional Approximate Counting\thanks{This work was supported by NSF Grant CCF-2221980.}}
\author{Dingyu Wang\\University of Michigan\\wangdy@umich.edu}

\date{}

\begin{document}
\maketitle

\begin{abstract}\small\baselineskip=9pt
The celebrated Morris counter uses $\log_2\log_2 n + O(\log_2 \sigma^{-1})$ bits to count up to $n$ with a relative error $\sigma$, where if $\hat{\lambda}$ is the estimate of the current count $\lambda$, then $\E|\hat{\lambda}-\lambda|^2 <\sigma^2\lambda^2$. The Morris counter was proved to be optimal in space complexity by Nelson and Yu \cite{nelson2022optimal}, even when considering the error tails. A natural generalization is \emph{multi-dimensional} approximate counting. Let $d\geq 1$ be the dimension. The count vector $x\in \N^d$ is incremented entry-wisely over a stream of coordinates $(w_1,\ldots,w_n)\in [d]^n$, where upon receiving $w_k\in[d]$, $x_{w_k}\gets x_{w_k}+1$. A \emph{$d$-dimensional approximate counter} is required to count $d$ coordinates simultaneously and return an estimate $\hat{x}$ of the count vector $x$. Aden-Ali, Han, Nelson, and Yu \cite{aden2022amortized} showed that the trivial solution of using $d$ Morris counters that track $d$ coordinates separately is already optimal in space, \emph{if each entry only allows error relative to itself}, i.e., $\E|\hat{x}_j-x_j|^2<\sigma^2|x_j|^2$ for each $j\in [d]$.  However, for another natural error metric---the \emph{Euclidean mean squared error} $\E |\hat{x}-x|^2$---we show that using $d$ separate Morris counters is sub-optimal. 

In this work, we present a simple and optimal $d$-dimensional counter with Euclidean relative error $\sigma$, i.e., $\E |\hat{x}-x|^2 <\sigma^2|x|^2$ where $|x|=\sqrt{\sum_{j=1}^d x_j^2}$, with a matching lower bound. We prove the following.
\begin{itemize}
    \item There exists a $(\log_2\log_2 n+O(d\log_2 \sigma^{-1}))$-bit $d$-dimensional counter with relative error $\sigma$.
    \item  Any $d$-dimensional counter with relative error $\sigma$ takes at least $(\log_2\log_2 n+\Omega(d\log_2 \sigma^{-1}))$ bits of space.
\end{itemize}
The upper and lower bounds are proved with ideas that are strikingly simple. The upper bound is constructed with a certain variable-length integer encoding and the lower bound is derived from a straightforward volumetric estimation of sphere covering. 
\end{abstract}

\section{Introduction}
In 1978, Morris \cite{morris1978counting} invented the classic \emph{approximate counter} which can count up to $n$ with $\log_2\log_2 n + O(\log_2 \sigma^{-1})$ bits, returning an estimate of the current count with a relative error $\sigma$. Such approximate counters are invented mainly to save space. However, it was noticed quite recently that approximate counters can also be much \emph{faster} on modern hardware in comparison to the deterministic counter due to their low \emph{write complexity} \cite{tristan2015efficient,steele2017adding,jayaram2024streaming}. 

% \begin{definition}[Morris counter]\label{def:morris}
%     For any $a>0$, the \emph{Morris counter} with parameter $a$ \cite{morris1978counting} stores an index $r$, initially 0. Upon an increment, with probability $(1+1/a)^{-r}$, $r\gets r+1$. Upon query, it returns $a((1+1/a)^r-1)$.
% \end{definition}

In the original paper \cite{morris1978counting}, Morris analyzed the mean and variance of the estimates. Flajolet \cite{flajolet1985approximate} analyzed the mean and variance of the \emph{index}\footnote{Roughly speaking, Morris counter stores an index $v$ to represent a count around $2^v$. Therefore, to count up to $n$, the index increases to $\log n$ and thus it takes $\log \log n$ to store the index.} of the Morris counter, from which a quantified space bound can be derived with Chebyshev's bound. Nelson and Yu \cite{nelson2022optimal} analyzed the \emph{tail} of the index, obtaining a sharper dependence on the failure probability. Specifically, for increments up to $\lambda$ and $\epsilon,\delta>0$, if an estimate $\hat{\lambda}$ is produced with $\pr(|\hat{\lambda}-\lambda|<\epsilon \lambda)\geq 1-\delta$ then the Morris counter needs $O(\log\log n + \log\epsilon^{-1}+\log\log \delta^{-1})$ space, and this is proved to be optimal  \cite{nelson2022optimal}.\footnote{Technically, one needs to keep a separate deterministic counter for small $\lambda=O(a)$ to obtain the optimal failure rate \cite{nelson2022optimal}.} 

It is the common scenario in real-world applications that many different counters are maintained simultaneously. For example, the original motivation of Morris is to count the number of each trigram in texts \cite{lumbroso2018story,nelson2022optimal} where $d=26^3$ different trigrams are counted simultaneously. Aden-Ali, Han, Nelson, and Yu \cite{aden2022amortized} showed that, essentially, $d$ independent counters cannot be compressed as long as each counter is required to report an estimate with an error \emph{relative to itself}. However, as we will later show, the space \emph{can be compressed} if the error metric is the $d$-dimensional Euclidean norm. While the per-entry relative error requirement considered in \cite{aden2022amortized} is reasonable  when an approximation for \emph{each individual entry} is needed, we remark that in applications where an approximation of the \emph{whole count vector} is needed, the Euclidean error metric is more natural. 

We now formally define the problem of $d$-dimensional counting in terms of Euclidean mean squared error. 
Let $(e_1,e_2,\ldots,e_d)$ be the standard basis of $\R^d$ where $e_k$ has its $k$th component being one and other components being zero. 

\begin{definition}[(Euclidean) $d$-dimensional counting]\label{def:counter}
Let $n\geq 1$ be the maximum stream length, $d\geq 1$ be the dimension, and $\sigma>0$ be the relative error.
An \emph{$(n,d,\sigma)$-counter} is a randomized data structure that, for any input sequence $(w_1,\ldots,w_k)\in [d]^k$ with $k\leq n$, returns an estimate $\hat{x}$ of $x=\sum_{j=1}^k e_{w_{j}}$ with $\E|\hat{x}-x|^2<\sigma^2|x|^2$, where $|x|=\sqrt{\sum_{j=1}^d x_j^2}$ is the Euclidean length. We denote the state space of the counter as $\Omega$ with $|\Omega|\geq 1$. One counter thus uses $\log_2 |\Omega|$ bits of space.
\end{definition}
Recall that the original Morris counter \cite{morris1978counting} with parameter $a\geq 1$, denoted as $\mathrm{Morris}(a)$, returns an unbiased estimate of the count with relative variance $O(1/a)$. Thus by definition,  $\mathrm{Morris}(a)$ is an $(n,1,O(\sqrt{1/a}))$-counter. The error here is measured by the \emph{mean squared error} $\E|\hat{x}-x|^2$, which is equivalent to measure the variance if the estimator is unbiased. We note that there are, of course, many other reasonable error metrics to consider, depending on the specific applications. The mean squared error is a good starting point for a new algorithmic problem due to its simple probabilistic/geometric structures. 

One natural way to construct a $d$-dimensional counter is to simply use $d$ separate Morris counters, which is the original method that Morris used to count trigrams \cite{lumbroso2018story}. Suppose now we have $d$ $(n,1,\sigma)$-counters to count $x_1,\ldots, x_d$. Let the estimates be $\widehat{x_1},\ldots,\widehat{x_d}$ respectively. Then we have
\begin{align*}
    \E |\hat{x}-x|^2 &= \E \sum_{j=1}^d |\widehat{x_j}-x_j|^2=\sum_{j=1}^d \E |\widehat{x_j}-x_j|^2
    \intertext{Since for any $j$, $\widehat{x_j}$ is an $(n,1,\sigma)$-counter for $x_j$, we have $\E |\widehat{x_j}-x_j|^2< \sigma^2x_j^2$. Thus,}
    \E |\hat{x}-x|^2  &< \sum_{j=1}^d \sigma^2x_j^2=\sigma^2|x|^2.
\end{align*}
We thus see $d$ $(n,1,\sigma)$-counters do form an $(n,d,\sigma)$-counter (Definition \ref{def:counter}). Each $(n,1,\sigma)$-counter can be implemented with a $\mathrm{Morris}(O(\sigma^{-2}))$ counter, taking $\log_2\log_2 n + O(\log_2\sigma^{-1})$ bits. Thus the total space needed is $d\log_2\log_2 n + O(d\log_2\sigma^{-1})$ bits, which is sub-optimal for this task. 

The main contribution of this work is the construction of a simple and optimal $(n,d,\sigma)$-counter, with a matching lower bound. 

\begin{theorem}\label{thm:main}
The following statements are true.
\begin{description}
    \item[Upper bound] 
For any $n\geq 2$, $d\geq 1$, and $\sigma\in(0,1/3)$, there exists an $(n,d,\sigma)$-counter with space size $\log_2|\Omega|=\log_2\log_2n+O(d\log_2\sigma^{-1})$. 
    \item[Lower bound] 
For any $n\geq e^2$, $d\geq 1$, and $\sigma\in(0,1/16)$, if there is an $(n,d,\sigma)$-counter with state space $\Omega$, then $\log_2|\Omega|\geq \log_2\log_2n+\Omega(d\log_2\sigma^{-1})$. 
\end{description}
\end{theorem}

A related data structure is the \emph{Count-Min sketch} \cite{cormode2005improved} by Cormode and Muthukrishnan. Given a count vector $x\in \N^d$ with $\sum_{j=1}^d x_j = n$, Count-Min is able to return an estimate $\widehat{x_j}$ of $x_j$, for any $j\in[d]$, such that $\widehat{x_j}\in[x_j,x_j+\epsilon n]$ with probability $1-\delta$, using only $O(\epsilon^{-1}\log\delta^{-1}\log n)$ bits of space. It may seem that Count-Min with the power of hash functions has a chance to beat the lower bound in Theorem \ref{thm:main} by combining the estimates for each coordinate. However, it is hopeless to use Count-Min as a $d$-dimensional counter unless $\epsilon=O(1/d)$, in which case it still lies above the lower bound in Theorem \ref{thm:main}. Indeed, suppose $\widehat{x_j}\in[x_j,x_j+\epsilon n]$ holds for all $j$. The vector estimate $\hat{x}=(\widehat{x_1},\ldots,\widehat{x_d})$ has error
\begin{align*}
    |\hat{x}-x|^2 \leq d\epsilon^2 n^2 \leq d^2\epsilon^2 |x|^2,
\end{align*}
where we use the bound $|x|^2\geq (\sum_{j=1}^d x_j)^2/d=n^2/d$. Therefore $\epsilon$ has to be $O(1/d)$  for $\hat{x}$ to be a multiplicative estimate of $x$. 

\subsection{Technical Overview}
\subsubsection{Upper Bound} 
The upper bound in Theorem \ref{thm:main} uses the natural idea of maintaining a common \emph{scale counter} $U$ for all $d$ dimensions and tracking the relative magnitude of each coordinate with a vector $V\in \N^d$. The estimate $\hat{x}$ is equal to $2^U V$. For demonstration, we will first describe a very simple but sub-optimal solution and then discuss how to modify it to reach optimal space-accuracy tradeoff. The transition of the states is designed similarly with the design of the Morris counter, in which the estimate is maintained \emph{unbiased} at any moment ($\E 2^U V =x$). Initially we set $U=0$ and $V=(0,\ldots,0)$.
\begin{itemize}
    \item When the $j$th coordinate is incremented (i.e., $x_j\gets x_j+1$), we update $V_j\gets V_j+1$ with probability $2^{-U}$, ensuring the unbiasedness of the estimate.
    \item Each coordinate $V_j$ is restricted to the range $[0,a]$, so that the space to store $V$ is $d\log_2(1+ a)$. Whenever some $V_j$ exceeds $a$, we have to \emph{scale up} to keep $V_j$ in range. Again, we want to maintain the property that $\E 2^UV= x$. Thus, the scale-up should be done as follows.
    \begin{itemize}
        \item $U\gets U+1$. The scale counter is increased by 1.
        \item For the estimate to be unbiased, $V_j$ should become $V_j/2$ for every $j$. Here we also want $V_j$ to always be an integer. Therefore, if $V_j$ is even, then $V_j\gets V_j/2$. If $V_j$ is odd, then $V_j\gets (V_j+\xi)/2$ where $\xi\in\{-1,1\}$ is a freshly sampled Rademacher random variable. 
    \end{itemize}
\end{itemize}
The algorithm above is unbiased by design. However, the memory-accuracy tradeoff is \emph{not} optimal. There is a fundamental problem of the design above that the parameter $a$ has to be $\Omega(\sqrt{d})$ (with $\sigma=O(1)$) for it to be a $d$-dimensional counter. Instead of analyzing the algorithm in details, this problem can be identified by simply checking the set of \emph{all possible estimates} that are produced by the algorithm:
\begin{align*}
    \mathcal{D}=\{2^U (V_1,\ldots,V_d): U\in \N ,\forall j\in[d], V_j\in[0,a]\}.
\end{align*}
It turns out that $ \mathcal{D}$ is not \emph{dense enough} for estimating some input $x$. Indeed, we choose 
\begin{align*}
x=2^s(\underbrace{a,\ldots,a}_{r},\underbrace{1/2,\ldots,1/2}_{d-r}),
\end{align*}
for some $r\in[0,d]$ so that both $y_1=2^s(a,\ldots,a,1,\ldots,1)$ and $y_2=2^{s-1}(a,\ldots,a,1,\ldots,1)$ are far from $x$.  By construction, the nearest estimate of $x$ in $\mathcal{D}$ has to be either $y_1$ or $y_2$.\footnote{Note that $y_1$ is one possible probabilistic rounding of $x$. Other roundings $\{2^s(a,\ldots,a,z_1,\ldots,z_{d-r})\mid z_j\in\{0,1\}\}$ all have the same distance to $x$. By symmetry, we only need to consider $y_1$. The vector $y_2$ is produced by decreasing the scale to $2^{s-1}$, so the smaller entries can be represented precisely with the cost of underestimating the large entries. }  We choose $r$ so that the two choices have equal distances to $x$. A simple calculation shows that we need to set
$r=d/(1+4a^2)$.\footnote{For simplicity, we only require $r>1$ and omit the details of rounding $r$ to an integer.} Assume now $a<\sqrt{d}/3$ and then $r>1$. Thus $\min_{y\in \mathcal{D}}|x-y|^2 = 2^{2s}\frac{da^2}{1+4a^2}$, while $|x|^2 = 2^{2s}\cdot 2a^2d/(1+4a^2)$. Therefore, if the algorithm is an $(n,d,\sigma)$-counter, then $\E |\hat{x}-x|^2<\sigma^2|x|^2$, which means there is at least one $x_*\in \mathcal{D}$ such that $|x_*-x|^2<\sigma^2|x|^2$. This suggests that $2^{2s}a^2d/(1+4a^2)< \sigma^2 2^{2s}\cdot 2a^2d/(1+4a^2)$ and therefore $\sigma>1/\sqrt{2}$. In other words, if $a<\sqrt{d}/3$, then any $(n,d,\sigma)$-counter must have $\sigma>1/\sqrt{2}$. We conclude that $a$ needs to be $\Omega(\sqrt{d})$ when $\sigma\leq 1/\sqrt{2}$, indicating an additional $O(d\log d)$ space overhead in comparison to the optimal space.

From the discussion above, we see that one has to have a dense enough set of estimates $\mathcal{D}$ for a correct $(n,d,\sigma)$-counter. 
Roughly speaking, one wants to spend more bits on large coordinates and fewer bits on small coordinates. This can be done with a simple and natural algorithmic idea: \emph{variable-length integer encoding}. The usual binary encoding of non-negative integers automatically uses more bits on large numbers and fewer bits on small numbers. We will prove the following set of estimates is dense enough.
\begin{align*}
    \mathcal{D}_{\text{compressed}}=\{2^U (V_1,\ldots,V_d): U\in \N, \text{the encoding length of $V$ is at most $O(d\log_2 \sigma^{-1})$}\},
\end{align*}
This leads to our upper bound in Theorem \ref{thm:main}. One may observe how this encoding idea better handles the previous counterexample $x=2^s(a,\ldots,a,1/2,\ldots,1/2)$. An estimate of $2^{s-1}(2a,\ldots,2a,1,\ldots,1)$ can now be returned because though each entry with value $2a$ needs above average space to encode, each entry with value $1$ uses less than average space to encode, \emph{saving the memory space for the large entries}. 
% Overall, if one compares $\mathcal{D}$ and $\mathcal{D}_{\text{compressed}}$ with the same space budget for the vector $V$, one notes that though, by construction, the total number of points are the same (say, for a bounded $U$), the points in $\mathcal{D}_{\text{compressed}}$ are better aligned with the task of approximating $x$ with relative Euclidean error. 

\subsubsection{Lower Bound}\label{sec:lower_intro}
As we have discussed above, the set of all estimates produced by the counter has to be dense enough for the task of $d$-dimensional counting to be possible. 
The lower bound in Theorem \ref{thm:main} thus arises from estimating the size of \emph{multiplicative space coverings}.

\begin{definition}[multiplicative space covering]\label{def:cover}
For any subsets $A,R\subset \R^d$ and $\sigma>0$, we say $R$ is a \emph{$\sigma$-multiplicative covering} of $A$ if for any $x\in A$ there exists $y\in R$ such that $|x-y|<\sigma |x|$.
\end{definition}
Clearly, the set of all estimates of an $(n,d,\sigma)$-counter will form a $\sigma$-multiplicative space covering of $\mathcal{X}=\{x\in\N^d\mid x_1+\cdots +x_d \leq n\}$. The lower bound is thus proved by estimating the covering size of $\mathcal{X}$, which can be further reduced to the classic problem of covering spheres with smaller spheres \cite{rogers1963covering}.

\subsection{Related Work}
There have been many variants and analysis of Morris counters before \cite{kruskal1991flexible,LOUCHARD2008109,gronemeier2009applying,fuchs2012approximate,nelson2022optimal}. Nevertheless, they have only considered the one-dimensional case. Tracking a multi-dimensional count vector (with a vector norm) is a new problem to the best of our knowledge. 

In the one-dimensional case ($d=1$), our algorithm is equivalent to the variant---\emph{floating-point counter}---by Cs{\H{u}}r{\"o}s \cite{csHuros2010approximate}, except for that the floating-point counter restricts $a$ to be a power of two so no probabilistic rounding is needed when halving. Nevertheless, when $d\geq 2$, probabilistic rounding is inevitable. In comparison to the classic Morris counter, the floating-point variant is much more convenient to implement on binary machines, where only simple integer operations and bit operations are needed \cite{csHuros2010approximate}. Another advantage is that the number of random bits needed per increment is at most 2 in expectation, since the probability $2^{-U}$ can be simulated by generating random bits sequentially, looking for $U$ consecutive ones \cite{csHuros2010approximate} (stop generating random bits when a zero shows up). Such advantages are inherited by our $d$-dimensional counter as well.

\subsection{Organization}
We will construct the upper bound in \S \ref{sec:upper} and prove the lower bound in \S \ref{sec:lower}. The main theorem (Theorem \ref{thm:main}) follows directly from Corollary \ref{cor:upper_bound} and \ref{cor:lower_bound}. We have included a sample run in Appendix \ref{sec:sample_runs} to help illustrate the patterns of the new algorithm. 

\section{Upper Bound: Variable-length Integer Encoding}\label{sec:upper}
We consider a specific integer encoding which is carefully designed to simplify the analysis later.
\begin{definition}[variable-length integer encoding]\label{def:code}
Let the symbol set be $\{0,1,|\}$ where $|$ is used as a separator. We encode $0$ as $|$, $1$ as $0|$. For $k\geq 2$, 
we encode $k$ as $[k-1]_2|$, where $[k-1]_2$ is the binary representation of $k-1$. See the following table. 
\begin{center}
\begin{tabular}{c|c|c|c}
   integer  & code & integer & code \\
   \hline
    $0$ &  $|$ & $4$ &  $11|$   \\
    $1$ & $0|$ & $5$ &  $100|$ \\
    $2$ & $1|$ & $6$ &  $101|$  \\
    $3$ & $10|$ & $7$ & $110|$
\end{tabular}
\end{center}
For any $k\in \N$, we define $\psi(k)$ to be the length of its encoding above. An integer vector $V\in\N^d$ is encoded by concatenating the codes for each of its component. For example $V=(3,0,4,0,1)$ is encoded as $10||11||0|$. The code length of the vector $V$ is denoted by $\psi(V)=\sum_{j=1}^d \psi(V_j)$.
\end{definition}
\begin{remark}
For clarity, the code here uses an alphabet of size $3$, including a separator. If it is encoded with bits, then the length will increase by a $\log_2 3$ factor. See \cite{patrascu2008succincter,dodis2010changing} for techniques that encode $m_*$ $\{0,1,|\}$-symbols using $\ceil{m_*\log_23}$ bits with constant querying time for each entry. 
\end{remark}

We now formally describe the new algorithm.
 The algorithm stores a \emph{scale counter} $U\in \N$ and a \emph{relative vector} $V=(V_1,\ldots,V_d)\in \N^d$. The estimator is $\hat{x}=2^UV=(2^UV_1,\ldots,2^UV_d)$. As discussed in the introduction, we will use the variable-length encoding in Definition \ref{def:code} to store the vector $V$. Let $M_*$ be the memory budget for $V$. Recall that $\psi(V)$ is the encoding length of $V$.

\medskip
\centerline{\framebox{
\begin{minipage}{5.5in}
\begin{description}
    \item[\texttt{Initialization}] $(V_1,\ldots,V_d)\gets(0,\ldots,0)$ and $U\gets 0$.
    \item[\texttt{Increment}$(j)$] With probability $2^{-U}$, $V_j\gets V_j+1$. If $\psi(V)>M_*$, then execute \texttt{Scale-up}.
    \item[\texttt{Scale-up}] $U\gets U+1$ and for $k\in \{1,\ldots,d\}$,
        \begin{itemize}
            \item if $V_k$ is even, then $V_k \gets V_k/2$;
            \item if $V_k$ is odd, then $V_k\gets (V_k+\xi)/2$ where $\xi\in\{-1,1\}$ is a freshly sampled Rademacher random variable.
        \end{itemize}
    \item[\texttt{Query}] Return $2^U(V_1,\ldots,V_d)$.
\end{description}
\end{minipage}
}}
\medskip

Since the algorithm will execute \texttt{Scale-up} whenever the space $\psi(V)$ of storing $V$ exceeds the budget $M_*$, the behavior of the algorithm will depend on the way that $V$ is encoded. The codes in Definition \ref{def:code} are constructed so that the space $\psi(V)$ increases by at most one during \texttt{Increment} and decreases by at least one during \texttt{Scale-up}. This guarantees that the memory space to encode $V$ is always bounded by $M_*$. We now list and prove all the properties of the encoding in Definition \ref{def:code} that we will use in the analysis. 
\begin{lemma}\label{lem:encoding}
For any $V\in \N^d$, the following statements are true.
\begin{itemize}
    \item $\psi(V)\leq  2d + \sum_{j=1}^d\log_2(1+V_j)$. (Space bound.)
    \item For any $j\in[d]$, $\psi(V+e_j) \leq \psi(V)+1$. (An increment increases the space by at most 1.)
    \item $\psi(\floor{V/2}) \geq \psi(V)-2d$. (Halving the relative vector will decrease the space by at most 2d. ) 
    \item If there exists $j\in[d]$ such that $V_j\geq 3$, then $ \psi(\ceil{V/2}) \leq \psi(V)-1$. (The space will decrease by at least one if not all $V_j$s are small.)
\end{itemize}
\end{lemma}
\begin{remark}
    $\floor{V/2}$ and $\ceil{V/2}$ are evaluated entry-wisely.
\end{remark}
\begin{proof}
Recall that for the usual binary encoding, an integer $k$ has length $1+\floor{\log_2k}$ if $k\geq 1$. Thus, by Definition \ref{def:code}, if $k\geq 2$, then $ 2+\floor{\log_2 (k-1)}$ symbols are needed to store the code $[k-1]_2|$. It requires one symbol to encode 0 and two symbols to encode 1. We bound the length for encoding $k$ by $2+\log_2(1+k)$, which holds for all $k\in \N$. We then have $\psi(V)=\sum_{j=1}^d \psi(V_j) \leq 2d+\sum_{j=1}^d \log_2(1+V_j)$. 

The second statement is true by construction.

For the third statement, it suffices to prove that $\psi(\floor{k/2})\geq \psi(k)-2$ for any $k\in \N$. One may check that this is true for $k=1,2,3,4$. Now assume $k\geq 5$. Note that $\psi(\floor{k/2})-\psi(k)=\floor{\log_2(\floor{k/2}-1)}-\floor{\log_2(k-1)}=-2+\floor{\log_2(4\floor{k/2}-4)}-\floor{\log_2(k-1)}$. It suffices to prove that $4\floor{k/2}-4\geq k-1$. Indeed, $4\floor{k/2}-4\geq 2k-6\geq k-1$ since $k\geq 5$. 

For the fourth statement, since $\norm{V}_\infty \geq 3$, there exists $j\in [d]$, such that $V_j\geq 3$. If $V_j$ is even, then $V_j\geq 4$ and $\psi(V_j/2)=2+\floor{\log_2(V_j/2-1)}= 1+\floor{\log_2(V_j -2)}\leq \psi(V_j)-1$. If $V_j$ is odd, then $\psi(\ceil{V_j/2})=\psi((V_j+1)/2)= 2+\floor{\log_2((V_j+1)/2-1)}= 1+\floor{\log_2(V_j-1)}=\psi(V_j)-1$. Note that $\psi$ is non-decreasing and thus we conclude that $\psi(\ceil{V/2}) \leq \psi(V)-1$.
\end{proof}

Next, we analyze the variance of the algorithm. Similar to the analysis of the Morris counter \cite{morris1978counting}, the basic idea is to analyze the \emph{change in variance} caused by each \texttt{Increment}, including a possible \texttt{Scale-up} step, and then compute the final variance by summing up the changes at time $k=1,\ldots,n$. 
Recall that
for any $r\in [d]$, $e_r$ is the $r$th standard base vector of $\R^d$. 
\begin{theorem}\label{thm:mean_var}
Fix any input sequence $(w_1,\ldots,w_n)\in [d]^n$.
Let $(U,V)$ be the final state of the algorithm and $x=\sum_{j=1}^n e_{w_j}$ be the final count vector. Define $\hat{x}=2^UV$. If the memory budget is $M_*=4d+2d \log_2(1+a)$ with $a\geq 1$, then 
\begin{align*}
\E \hat{x} & = x,\\
\E |\hat{x}-x|^2 & \leq \frac{5}{6a-2}|x|^2.
\end{align*}
Furthermore for any $r\geq 1$, 
\begin{align*}
   \pr\left(U \geq r+\log_2\left(\frac{n}{ad}+1\right)\right) &\leq 2^{-r}.
\end{align*}
\end{theorem}
\begin{proof}
$\E \hat{x} = x$ is true by construction, since both \texttt{Increment} and \texttt{Scale-up} are designed to keep $\hat{x}=2^UV$ unbiased. The challenge is to analyze the variance $\E |\hat{x}-x|^2$. We start by proving the following lemma that bounds the final variance by the distributional properties of the scale counter $U$ at time $k=1,\ldots,n$. 
\begin{lemma}[variance decomposition]\label{lem:var_decomp}
For any $k\leq n$, let $(U^{(k)},V^{(k)})$ be the memory state at time $k$ (i.e., after seeing $w_1,\ldots,w_k$) and $x^{(k)}=\sum_{j=1}^k e_{w_j}$ be the count vector at time $k$. Let $\hat{x}^{(k)}=2^{U^{(k)}}V^{(k)}$ to be the estimate at time $k$. Then
\begin{align*}
     \E|\hat{x}^{(n)}|^2-|x^{(n)}|^2 &\leq \sum_{k=0}^{n-1}\E(2^{U^{(k)}}-1)+\frac{d}{3}\E (2^{2U^{(n)}}-1).
\end{align*}
\end{lemma}
\begin{proof}
     We define $Y^{(k)}$ to be the intermediate vector after \texttt{Increment} before \texttt{Scale-up} at time $k$. By construction, we know $Y^{(k)}=V^{(k)}+e_{w_{k+1}}$ with probability $2^{-U^{(k)}}$; otherwise $Y^{(k)}=V^{(k)}$. For simplicity, we now fix $k\geq 0$ and write $U^{(k)},V^{(k)},x^{(k)},w_{k+1},Y^{(k)}$ as $U,V,x,w,Y$ and $U^{(k+1)},V^{(k+1)},x^{(k+1)}$ as $U',V',x'$ respectively. By the definition of \texttt{Increment}, we have
\begin{align}
    \E (|2^{U}Y|^2-|2^{U}V|^2 \mid U,V) &=2^{-U}\cdot 2^{2U}(2V_{w}+1)=2^U(2V_w+1).\label{eq:phase1}
\end{align}
Now if $\psi(Y)\leq M_*$, then we have $(U',V')=(U,Y)$ where no \texttt{Scale-up} is executed. If $\psi(Y)> M_*$, then we must have $\psi(Y)=M_*+1$ because by Lemma \ref{lem:encoding}, $\psi(Y)\leq \psi(V+e_{w})\leq \psi(V)+1$ and $\psi(V)\leq M_*$. By construction, $M_*\geq 3d$ and thus there exists $j\in[d]$ such that $Y_j\geq 3$. By Lemma \ref{lem:encoding},  we know $\psi(\ceil{Y/2})\leq \psi(Y)-1=M_*$. Therefore, after the \texttt{Scale-up} step, the space will be at most $M_*$. When \texttt{Scale-up} is executed, $U'=U+1$ and $V'=Y/2+\sum_{k=1}^d \xi_k \ind{2\nmid Y_k}$. Here, $\xi_1,\ldots,\xi_d$ are i.i.d.~Rademacher random variables. Thus we have
\begin{align}
    &\E (|2^{U'}V'|^2-|2^UY|^2\mid U,Y) \nonumber\\
    &=\ind{\psi(Y)>M_*}\sum_{k=1}^d \E \left[\left(2^{2U+2}\left(\frac{Y_k+\xi_k\ind{2\nmid Y_k}}{2}\right)^2 - 2^{2U}Y_k^2 \right)\;\middle|\; U,Y\right],\nonumber
    \intertext{note that $\E \xi_k=0$ and $\xi_k^2=1$}
    &=\ind{\psi(Y)>M_*}\sum_{k=1}^d \left(2^{2U}(Y_k^2+\ind{2\nmid Y_k}) - 2^{2U}Y_k^2 \right)\nonumber\\
    &=\ind{\psi(Y)>M_*}2^{2U}\sum_{k=1}^d \ind{2\nmid Y_k}\nonumber\\
    &\leq \ind{\psi(Y)>M_*}2^{2U}\cdot d.\label{eq:phase2}
\end{align}
Take the total expectation and by Equation (\ref{eq:phase1}) and (\ref{eq:phase2}) we have,
\begin{align}
    \E(|2^{U'}V'|^2-|2^{U}V|^2) &=  \E \E(|2^{U}Y|^2-|2^{U}V|^2\mid U,V)+\E\E(|2^{U'}V'|^2-|2^{U}Y|^2\mid U,Y)\nonumber\\    
    &\leq \E (2^U(2V_w+1))+\E(\ind{\psi(Y)>M_*}2^{2U}d).\label{eq:phase3}
\end{align}
Recall that $|x'|^2-|x|^2 = (x_w+1)^2-x_w^2= 2x_{w}+1$ and $\E 2^U V_w = x_w$. This implies $\E 2^U V_w = (|x'|^2-|x|^2-1)/2$. Thus by Equation (\ref{eq:phase3}) we have,
\begin{align}
    \E(|2^{U'}V'|^2-|2^{U}V|^2) &\leq  |x'|^2-|x|^2+\E (2^U-1) +\E( d 2^{2U} \ind{\psi(Y)>M_*}).\label{eq:phase4}
\end{align}
Sum up the differences and we have
\begin{align*}    
\E|2^{U^{(n)}}V^{(n)}|^2
    &=\sum_{k=0}^{n-1}\E(|2^{U^{(k+1)}}V^{(k+1)}|^2-|2^{U^{(k)}}V^{(k)}|^2)
    \intertext{Recall that $U',V',U,V$ are short for $U^{(k+1)},V^{(k+1)},U^{(k)},V^{(k)}$. By Equation (\ref{eq:phase4}),}
    &\leq \sum_{k=0}^{n-1}\left(|x^{(k+1)}|^2-|x^{(k)}|^2+\E(2^{U^{(k)}}-1)+\E \left(d2^{2U^{(k)}}\ind{\psi(Y^{(k)})>M_*}\right) \right)\\
    &=|x^{(n)}|^2+\sum_{k=0}^{n-1}\E(2^{U^{(k)}}-1) + d\E\left[\sum_{k=0}^{n-1}2^{2U^{(k)}}\ind{\psi(Y^{(k)})>M_*}\right].
\end{align*}
Note that $\ind{\psi(Y^{(k)})>M_*}$ is the indicator whether \texttt{Scale-up} is executed when processing $w_{k+1}$. Therefore 
\begin{align*}
    \sum_{k=0}^{n-1}2^{2U^{(k)}}\ind{\psi(Y^{(k)})>M_*}&=\sum_{j=0}^{U^{(n)}-1} 2^{2j} = \frac{4^{U^{(n)}}-1}{4-1}= \frac{2^{2U^{(n)}}-1}{3}
\end{align*}
We conclude that
\begin{align*}
    \E |2^{U^{(n)}}V^{(n)}|^2-|x^{(n)}|^2&\leq \sum_{k=0}^{n-1}\E(2^{U^{(k)}}-1)+\frac{d}{3}\E (2^{2U^{(n)}}-1),
\end{align*}
which gives the stated variance decomposition.
\end{proof}

The next step is to estimate $\E (2^{U}-1)$ and $\E (2^{2U}-1)$. The idea is to relate them with $\E 2^{U}V=x$ and $\E |2^UV|^2=\E |\hat{x}|^2$. 
\begin{lemma}
If $U\geq 1$, then $|V|^2 \geq \sum_{j=1}^d V_j \geq ad$.
\end{lemma}
\begin{proof}
    Since $V_j$s are integer, we automatically have $|V|^2=\sum_{j=1}^d V_j^2\geq \sum_{j=1}^d V_j$. Thus it suffices to prove $\sum_{j=1}^d V_j\geq ad$. Note that if $U\geq 1$, then the algorithm has executed at least one \texttt{Scale-up}. Note that by Lemma \ref{lem:encoding}, the space $\psi(V)$ increases at most by one per \texttt{Increment} and decreases at most by $2d$ per \texttt{Scale-up}. Therefore, after the first \texttt{Scale-up}, the memory usage $\psi(V)$ is between $M_*+1-2d$ and $M_*$. Note that $M_*$ is set to $4d+d\log_2(1+a)$. Thus we have 
    \begin{align*}
        \psi(V)&\geq M_*+1-2d \geq 2d +d\log_2(1+a) .
        \intertext{
    On the other hand, by Lemma \ref{lem:encoding} (encoding space bound), we have}
        \psi(V)&\leq 2d+ \sum_{j=1}^d\log_2(1+V_j).
    \end{align*}
    The bounds above imply that $
        \sum_{j=1}^d\log_2(1+V_j)\geq d\log_2(1+a)$. Since $\log_2(1+r)$ is convex in $r$, we conclude that
        \begin{align*}
            d\log_2(1+a)&\leq \sum_{j=1}^d\log_2(1+V_j) \leq d\log_2\left(1+\sum_{j=1}^dV_j/d\right),
        \end{align*}
        and thus $\sum_{j=1}^dV_j\geq ad$.
\end{proof}
Note that if $U=0$, then $2^U-1=0$.
With the lemma above, we now have
\begin{align*}
    ad\E (2^U-1) \leq \E (2^U-1)\sum_{j=1}^dV_j  \leq \E 2^U\sum_{j=1}^dV_j =\sum_{j=1}^dx_j=k,
\end{align*}
where $k$ is the current time, i.e., the total number of increments so far. This implies $\E (2^U-1)\leq k/(ad)$. Similarly, we have
\begin{align*}
    ad\E (2^{2U}-1) \leq \E (2^{2U}-1)\sum_{j=1}^dV_j^2 \leq \E 2^{2U}\sum_{j=1}^dV_j^2 =\E |\hat{x}|^2,
\end{align*}
which implies $\E (2^{2U}-1)\leq \E |\hat{x}|^2/(ad)$.  Inserting the estimates back into Lemma \ref{lem:var_decomp}, we have
\begin{align*}
    \E |\hat{x}^{(n)}|^2-|x^{(n)}|^2&\leq \sum_{k=0}^{n-1}\E(2^{U^{(k)}}-1)+\frac{d}{3}\E (2^{2U^{(n)}}-1)\\
    &\leq \sum_{k=0}^{n-1} k/(ad) + \frac{d}{3}\E |\hat{x}^{(n)}|^2/(ad)
    \intertext{Note that $\sum_{k=0}^{n-1} k\leq n^2/2\leq d|x^{(n)}|^2/2$ since $\sum_{j=1}^d x^{(n)}_j=n$ and $d|x^{(n)}|^2 \geq (\sum_{j=1}^d x^{(n)}_j)^2$. }
    &\leq \frac{1}{2a}|x^{(n)}|^2 + \frac{1}{3a}\E |\hat{x}^{(n)}|^2.
\end{align*}
Reorganizing the inequality above, we have
\begin{align*}
    \E|\hat{x}^{(n)}|^2-|x^{(n)}|^2 &\leq \frac{5}{6a-2}|x^{(n)}|^2.
\end{align*}
Note that since the estimate is unbiased, we have $\E|\hat{x}^{(n)}|^2-|x^{(n)}|^2=\E|\hat{x}^{(n)}-x^{(n)}|^2 $. Thus we conclude that $\E|\hat{x}^{(n)}-x^{(n)}|^2\leq \frac{5}{6a-2}|x^{(n)}|^2$. 

Finally we estimate the tail of $U$. We already know that  $\E(2^{U^{(k)}}-1)\leq k/(ad)$ and thus $\E 2^{U^{(k)}}\leq 1+k/(ad)$. Then for any $z>0$, by Markov,
\begin{align*}
    \pr(2^{U^{(k)}} \geq z) &\leq \frac{\E 2^{U^{(k)}}}{z} \leq \left(\frac{k}{ad}+1\right)/z
\end{align*}
Setting $z=2^r(k/(ad)+1)$ and we have
\begin{align*}
    \pr\left(U^{(k)} \geq r+\log_2\left(\frac{k}{ad}+1\right)\right) &\leq 2^{-r}.
\end{align*}
We thus have proved the tail bound for $U$.
\end{proof}

Theorem \ref{thm:mean_var} characterizes the mean and variance of the estimator. We now analyze the space. The vector $V$ is encoded with length at most $M_*$ by construction. We are left to analyze the memory space needed for the scale counter $U$. 
\begin{corollary}\label{cor:upper_bound}
For $\sigma\in(0,1/3)$, there exists an $(n,d,\sigma)$-counter with state space size $\log_2|\Omega|=\log_2\log_2 n + O(d\log_2\sigma^{-1})$.
\end{corollary}
\begin{proof}
If $n\leq \sigma^{-1}$, we may use $d$ deterministic counters taking $d\log_2 n \leq d\log_2\sigma^{-1}$ space which surely form a $(n,d,\sigma)$-counter. We thus assume $n>\sigma^{-1}$.

With a fixed space budget for storing the scale counter $U$, there is a chance that $U$ grows too large to be stored. 
In such case, the algorithm simply enters a fail state and returns $(0,\ldots,0)$ upon query. For $\tau>0$, we store $U$ only up to $U_* = \log_2(\tau^{-1}) +\log_2(n/(ad)+1)$. By design $U$ is non-decreasing in time and it suffices to analyze $U$ at the end of the stream. By Theorem \ref{thm:mean_var}, we know $
\pr(U\geq U_*)\leq  \tau.$ The estimate $\hat{x}$ can be written as $\ind{U<U_*}2^UV$. Its mean squared error is
\begin{align*}
    \E \left|\ind{U<U_*}2^{U}V-x\right|^2 &=\E\left( \ind{U<U_*}|2^{U}V-x|^2+\ind{U\geq U_*}|x|^2\right)\\
    &\leq \E |2^{U}V-x|^2 + \pr(U\geq U_*)|x|^2
    \intertext{Note that $ \E |2^{U}V-x|^2\leq \frac{5}{6a-2}|x|^2$ by Theorem \ref{thm:mean_var} and $\pr(U\geq U_*)\leq \tau$ by construction.}
    &\leq \left(\frac{5}{6a-2}+\tau\right)|x|^2 
\end{align*}
We set $\tau = \sigma^2/2$ and $a=2\sigma^{-2}$ and have
\begin{align*}
    \frac{5}{6a-2}+\tau =\frac{5}{12\sigma^{-2}-2}+\sigma^2/2 \leq \sigma^2,
\end{align*}
since $\sigma<1$. In this setting, we have $\E | \ind{U<U_*}2^{U}V-x|^2\leq \sigma^2|x|^2$ which implies the algorithm is an $(n,d,\sigma)$-counter. We now compute its space usage. By construction, the space used by $U$ is at most $\log_2(\log_2(\tau^{-1}) +\log_2(n/(da)+1)))\leq \log_2\log_2n+O(1)$ bits since we assumed $\tau^{-1}<n$. The space used by the vector $V$ is by construction $M_*=4d+d\log_2(1+a)=O(d\log \sigma^{-1})$.  Note that $V$ is encoded with a alphabet set of size $3$ ($\{0,1,|\}$) and $M_*$ is the maximum number of symbols. One may encode $V$ using bits by adding an additional $\log_2 3$ factor to the leading constant. We conclude that the total space needed to store $(U,V)$ is $\log_2\log_2 n + O(d\log\sigma^{-1})$.
\end{proof}

\section{Lower Bound: Multiplicative Space Covering}\label{sec:lower}
Recall that by Definition \ref{def:cover}, 
for any subset $A,R\subset \R^d$ and $\sigma>0$, we say $R$ is an $\sigma$-multiplicative covering of $A$ if for any $x\in A$ there exists $y\in R$ such that $|x-y|<\sigma |x|$.

We now prove that any correct $(n,d,\sigma)$-counter will induce an $O(\sigma)$-multiplicative covering of the space. 
The proof here will be slightly more complicated than what we described in \S \ref{sec:lower_intro}. Previously we assumed the algorithm returns a fixed estimate for a given state and therefore the set of estimates naturally form a covering. In general, a probabilistic algorithm is allowed to produce a randomized answer even if the final memory state is fixed. We show that even when estimates are randomized for a given state, an $(n,d,\sigma)$-counter will still induce a covering of the space.
\begin{lemma}\label{lem:integer_to_real}
If $\Omega$ is the state space of an $(n,d,\sigma)$-counter with $\sigma<1/3$, then there exists a  $4\sigma$-multiplicative covering $R$ of $\{x\in\R_+^d \mid \sqrt{d}\sigma^{-1}\leq|x|\leq  n/\sqrt{d}\}$ with $|R|\leq |\Omega|$.
\end{lemma}
\begin{proof}
Let $g:\Omega\to \R^d$ be the randomized function that returns the estimate $\hat{x}=g(S)$ based on the final state $S\in \Omega$. Let $\mathcal{X}=\{x\in \N^d \mid \sum_{j=1}^d x_j \leq n\}$ be the space of all possible count vectors. By definition of the $(n,d,\sigma)$-counter, for any $x\in \mathcal{X}$, we have
\begin{align*}
    \E |g(S)-x|^2 < \sigma^2 |x|^2.
\end{align*}
For any $x\in \mathcal{X}$, define $s_x=\arg\min_{s\in \Omega}  \E |g(s)-x|^2 $. Thus we have
\begin{align}
   \E |g(s_x)-x|^2  \leq \E |g(S)-x|^2 <\sigma^2 |x|^2. \label{eq:rand1}
\end{align}
This implies  $\E |g(s_x)-x| \leq \sqrt{\E |g(s_x)-x|^2}<\sigma|x|$.
Suppose for some $y\in \mathcal{X}$ with $s_y=s_x$. By using the same randomness for $g(s_x)$ and $g(s_y)$, we have
\begin{align}
    |x-y| &= |g(s_x)-x-(g(s_y)-y)| \leq |g(s_x)-x|+|(g(s_y)-y)|\label{eq:rand2}\\
    \intertext{By Equation (\ref{eq:rand1}), we know $\E |g(s_x)-x|<\sigma|x|^2$ and $\E |g(s_y)-y|<\sigma|y|^2$. Take expectation on both sides of Equation (\ref{eq:rand2}) and we have }    
    |x-y| &\leq \E |g(s_x)-x|+\E |(g(s_y)-y)|
     < \sigma(|x|+|y|).\label{eq:gap}
\end{align}
We now group the count vectors $x$ by the value $s_x$. For any $s\in\Omega$, let $A_s=\{x\in \mathcal{X} \mid s_x=s\}$ and $\{A_s|s\in \Omega\}$ form a partition of $\mathcal{X}$. For each $s$ that $A_s\neq \emptyset$, we choose a representation $x_*^{(s)}\in A_s$ according to the lexicographic order. Next we are going to show that $R=\{x_*^{(s)} \mid s\in \Omega,A_s\neq\emptyset\}$ forms a good multiplicative covering of $\{x\in\R_+^d \mid d\sigma^{-1}<|x|< n/d\}$.

For any $x\in \R_+^d$, denote $\floor{x}=(\floor{x_1},\ldots,\floor{x_d})$. We have $|x-\floor{x}|<\sqrt{d}$. Note that
\begin{align*}
    \sum_{j=1}^d \floor{x_j} \leq \sqrt{d}|\floor{x}| \leq \sqrt{d}|x|.
\end{align*}
If $|x|\leq n/\sqrt{d}$, then $ \sum_{j=1}^d \floor{x_j}\leq \sqrt{d}|x|\leq n$, which implies that $\floor{x}\in \mathcal{X}$. Set $y=x_*^{(s_{\floor{x}})}\in R$. Since $s_{y}=s_{\floor{x}}$, we have by Equation (\ref{eq:gap}),
\begin{align*}
    |y-\floor{x}|<\sigma(|y|+|\floor{x}|)\leq \sigma(|\floor{x}|+|y-\floor{x}|+|\floor{x}|)
\end{align*}
Reorganizing,
\begin{align*}
    |y-\floor{x}|<\frac{2\sigma}{1-\sigma} |\floor{x}|.
\end{align*}
Thus we have
\begin{align*}
    |y-x| &\leq |y-\floor{x}|+|x-\floor{x}|\\
    &<\frac{2\sigma}{1-\sigma} |\floor{x}|+\sqrt{d}\\
    &\leq \frac{2\sigma}{1-\sigma} |x|+\sqrt{d},
\end{align*}
If $|x|\geq\sigma^{-1}\sqrt{d}$, then $\sqrt{d}\leq\sigma|x|$. Thus we have
\begin{align*}
    |y-x|< \frac{2\sigma}{1-\sigma} |x|+\sigma|x| =\frac{3\sigma-\sigma^2}{1-\sigma}|x|.
\end{align*}
Since $\sigma<1/3$, we have $\frac{3\sigma-\sigma^2}{1-\sigma}\leq 4\sigma.$
Thus $|y-x|<4\sigma|x|$. Note that the above argument holds for any $x$ with $|x|\in [\sigma^{-1}\sqrt{d},n/\sqrt{d}]$. Therefore $R$ is a $4\sigma$-multiplicative cover of $\{x\in\R_+^d \mid \sigma^{-1}\sqrt{d}\leq |x|\leq n/\sqrt{d}\}$. Finally, note that the covering $R= \{x_*^{(s)}\mid s\in\Omega,A_s\neq\emptyset\}$ has size at most $|\Omega|$.
\end{proof}

Lemma \ref{lem:integer_to_real} shows that the state space of an $(n,d,\sigma)$-counter will be at least the size of an optimal $4\sigma$-multiplicative covering of the shell\footnote{Technically, the shell intersected with the positive cone.} $\{x\in\R_+^d \mid \sqrt{d}\sigma^{-1}\leq|x|\leq  n/\sqrt{d}\}$. We now bound the size of the shell covering. The first step is to bound the size of the sphere covering. 
Covering a sphere with smaller spheres is a classic problem \cite{rogers1963covering}. We prove the following lemma based on \cite{boroczky2003covering}.
\begin{lemma}\label{lem:cover_sphere}
For any $\sigma\in(0,1/3)$,
if $R\subset \R^d$ is an $\sigma$-multiplicative cover of $\{x\in\R_+^d\mid |x|=1\}$, then $|R|\geq 2^{-d}\sigma^{-(d-1)}$.
\end{lemma}
\begin{proof}
We first consider the case $d=1$. The target set $\{x\in\R_+\mid |x|=1\}=\{1\}$ has size 1 and thus $|R|= 1 >2^{-1}$. We now assume $d\geq 2$.

Let $\mathbb{S}^{d-1}$ be the unit sphere in $\R^d$. Let $|\mathbb{S}^{d-1}|$ be the surface area of the sphere. Let the maximum area on $\mathbb{S}^{d-1}$ that can be covered by a single ball with radius $\sigma$ be $\eta_\sigma$. By Corollary 3.2 in \cite{boroczky2003covering}, 
\begin{align*}
    \frac{|\mathbb{S}^{d-1}|}{\eta_\sigma} \geq \sqrt{2\pi( d-1)} \frac{\sqrt{1-\sigma^2}}{\sigma^{(d-1)}},
\end{align*}
if $\sqrt{1-\sigma^2} \geq 1/\sqrt{d}$, which is true since we  assumed $\sigma<1/3$ and $d\geq 2$.

On the other hand, by symmetry, $|\mathbb{S}^{d-1}\cap \R_+^d|=2^{-d}|\mathbb{S}^{d-1}|$. Therefore, the number of points in $R$ is at least
\begin{align*}
  \frac{|\mathbb{S}^{d-1}\cap \R_+^d|}{\eta_\sigma }  &\geq  2^{-d}\sqrt{2\pi( d-1)} \frac{\sqrt{1-\sigma^2}}{\sigma^{(d-1)}}\geq 2^{-d}\sigma^{-(d-1)},
\end{align*}
since $d\geq 2$ and $\sqrt{2\pi(1-1/9)}\geq 1$.
\end{proof}

With the bound for the sphere covering in hand, the next step is to count the number of \emph{layers} in the shell so that the coverings for different layers are disjoint. Define $\mathbb{S}_+^{d-1}=\mathbb{S}^{d-1}\cap \R_+^d$, $\alpha\mathbb{S}_+^{d-1}=\{x\in\R_+^d \mid |x|=\alpha\}$, and $[\alpha,\beta ]\mathbb{S}_+^{d-1}=\{x\in\R_+^d \mid |x|\in[\alpha,\beta]\}$.

\begin{corollary}\label{cor:cover_shell}
For $0<\alpha<\beta$ and $\sigma\in(0,1/3)$, if $R$ is an $\sigma$-multiplicative cover of $[\alpha,\beta]\mathbb{S}_+^{d-1}$, then $|R|\geq \frac{1}{3}2^{-d}\sigma^{-d}\log(\beta/(e\alpha))$.
\end{corollary}
\begin{proof}
For any $x\in \R^d$, if $|x-y|<\sigma |x|$, then $|y|\in ((1-\sigma)|x|,(1+\sigma)|x|)$. Now define $b_k=(1+3\sigma)^k$ for $k\in \N$. Define $$R_k=R\cap \{x\in\R_+^d : |x| \in ((1-\sigma)b_k,(1+\sigma)b_k)\}$$ for any $k\in \N$. It is straightforward to check that $b_k(1+\sigma) < b_{k+1}(1-\sigma)$ since $\sigma<1/3$. Therefore $R_k$s are disjoint for different $k$s.  By construction, if $b_k \in [\alpha,\beta]$, then $R_k$  covers $b_k\mathbb{S}_+^{d-1}$. By Lemma \ref{lem:cover_sphere}, we know $|R_k|\geq 2^{-d}\sigma^{-(d-1)}$ (the \emph{multiplicative} covering of a sphere does not depend on the radius.) Thus 
\begin{align*}
    |R|\geq \sum_{k:b_k\in[\alpha,\beta]} |R_k| \geq 2^{-d}\sigma^{-(d-1)}\sum_{k:b_k\in[\alpha,\beta]}1,
\end{align*}
and we only need to count the number of $b_k$s that are in $[\alpha,\beta]$.
Note that
\begin{align*}
    b_k\in [\alpha,\beta] \iff k\log(1+3\sigma)\in[\log \alpha,\log \beta]\iff k\in[\frac{\log \alpha}{\log(1+3\sigma)},\frac{\log \beta}{\log(1+3\sigma)}].
\end{align*}
We conclude that 
\begin{align*}
    |R|&\geq  2^{-d}\sigma^{-(d-1)}(\frac{\log \beta-\log \alpha}{\log(1+3\sigma)}-1)
\intertext{relaxing $1+3\sigma< e^{3\sigma}$ and $1\leq 1/(3\sigma)$}
&> 2^{-d}\sigma^{-(d-1)}(\frac{1}{3\sigma}\log (\beta/\alpha)-\frac{1}{3\sigma})\\
&= \frac{1}{3}2^{-d}\sigma^{-d}\log(\beta/(e\alpha)),
\end{align*}
which lower bounds the size of $\sigma$-multiplicative covers of $[\alpha,\beta]\mathbb{S}_+^{d-1}$.
\end{proof}

We now finish the proof of the space lower bound of $(n,d,\sigma)$-counters.
\begin{corollary}\label{cor:lower_bound}
If $\Omega$ is the state space of an $(n,d,\sigma)$-counter with $\sigma<1/3$, 
then 
\begin{align*}
    |\Omega|\geq \frac{1}{3}2^{-d}(4\sigma)^{-d}\log( n/(e \sigma^{-1} d)).
    \intertext{This implies that for $n\geq e^2$, $d\geq 1$, and $\sigma<1/16$,}
    \log_2|\Omega| = \log_2\log_2 n + \Omega(d\log \sigma^{-1}).
\end{align*}
\end{corollary}
\begin{proof}
By Lemma \ref{lem:integer_to_real}, we know there exists a $4\sigma$-multiplicative covering $R$ of $[\sqrt{d}\sigma^{-1},n/\sqrt{d}]$ with $|R|\leq |\Omega|$. By Corollary \ref{cor:cover_shell}, we know $|R|\geq \frac{1}{3}2^{-d}(4\sigma)^{-d}\log(n/(ed\sigma^{-1}))$. This proves the first statement. For the second statement, note that
\begin{align*}
    \log_2|\Omega| &\geq d\log_2\sigma^{-1}-3d-\log_2 3 + \log_2\log (n/(e\sigma^{-1}d)).
    \intertext{If $\sigma<1/16$ then $\log_2\sigma^{-1}>4$ and thus $3\leq \frac{3}{4}\log_2\sigma^{-1}$. This implies that $d\log_2\sigma^{-1}-3d\geq \frac{1}{4}d\log_2\sigma^{-1}$. Furthermore, if $\sigma^{-1}\geq n^{1/6}$ or $d\geq n^{1/6}$, then $\log_2\log_2 n =o(d\log \sigma^{-1})$. Otherwise, we have $n/(d\sigma^{-1})\geq n^{1/3}$. Note that $e\leq n^{1/2}$ by assumption. Thus we have $ \log_2\log (n/(e\sigma^{-1}d)) \geq \log_2\log n^{1/6} = \log_2\log_2 n + \log_2 (\frac{1}{6}\log 2)$. We conclude that }
    \log_2|\Omega| &= \log_2\log_2 n + \Omega(d\log \sigma^{-1}) - O (1)\\
    &= \log_2\log_2 n + \Omega(d\log \sigma^{-1}),
\end{align*}
when $n\geq e^2$, $d\geq 1$, and $\sigma<1/16$.
\end{proof}

% \section{Conclusion}
% In this work we have developed an optimal algorithm for $d$-dimensional counting in terms of  Euclidean mean square error, indicating that tracking a $d$-dimensional count vector over an incremental stream (the online problem) has the same space complexity with representing a given count vector in a sparse grid up to multiplicative error (the offline problem). 

% Unlike most algorithms where compression is used only as an add-on and does not affect the algorithm's execution, the $d$-dimensional counter here will crucially depend on the way the vector is compressed. The set of all estimates
% \begin{align*}
%     \left\{2^U(V_1,\ldots,V_d) : U\in \N, V\in \N^d, \text{ the encoding length of $(V_1,\ldots,V_d)$ is less than $M_*$}\right\}
% \end{align*}
% thus present an \emph{algorithmic} multiplicative covering of the space $\{x\in\N^d:\sum_{j=1}^dx_j\leq n\}$. 

\bibliographystyle{plain}
\bibliography{refs}

\appendix
\section{Sample Runs}\label{sec:sample_runs}
For demonstration, we run the $d$-dimension counter defined in \S \ref{sec:upper} with $d=4$ and memory budget $M_*=12$. The input is generated randomly at each step with probability distribution $(1/2,1/4,1/8,1/8)$. The table below tracks the state of the algorithm whenever $(U,V)$ is changed. Recall that $x$ is the input count vector, the estimator is $\hat{x}=2^UV$, and the encoding of $V$ is defined in Definition \ref{def:code}. We make the following remarks.
\begin{itemize}
    \item When $U=0$, $V$ tracks $x$ deterministically. 
    \item When $x$ is incremented, the encoding length $\psi(V)$ of $V$ increases by at most one. When $\psi(V)$ reaches $M_*+1$, the encoding length decreases by at least 1 and at most $2d$ during the \texttt{Scale-up}. when $U\geq 1$, i.e., after the first \texttt{Scale-up}, the encoding length is between $M_*+1-2d$ and $M_*$. These are the key properties we have used in the proof of Theorem \ref{thm:mean_var}.
    \item Since the rates for inserting the entries are skewed, the space used to encode each entry is different. Larger entries are expected to take more space. 
\end{itemize}
\begin{verbatim}
  U  |     V     |       estimate        |           x           |  encoded V
-----+-----------+-----------------------+-----------------------+---------------
  0  |  0 0 0 0  |     0    0    0    0  |     0    0    0    0  |  ||||
  0  |  0 1 0 0  |     0    1    0    0  |     0    1    0    0  |  |0|||
  0  |  0 1 0 1  |     0    1    0    1  |     0    1    0    1  |  |0||0|
  0  |  1 1 0 1  |     1    1    0    1  |     1    1    0    1  |  0|0||0|
  0  |  2 1 0 1  |     2    1    0    1  |     2    1    0    1  |  1|0||0|
  0  |  2 1 0 2  |     2    1    0    2  |     2    1    0    2  |  1|0||1|
  0  |  2 2 0 2  |     2    2    0    2  |     2    2    0    2  |  1|1||1|
  0  |  3 2 0 2  |     3    2    0    2  |     3    2    0    2  |  10|1||1|
  0  |  4 2 0 2  |     4    2    0    2  |     4    2    0    2  |  11|1||1|
  0  |  4 3 0 2  |     4    3    0    2  |     4    3    0    2  |  11|10||1|
  0  |  5 3 0 2  |     5    3    0    2  |     5    3    0    2  |  100|10||1|
  0  |  5 3 1 2  |     5    3    1    2  |     5    3    1    2  |  100|10|0|1|
  0  |  6 3 1 2  |     6    3    1    2  |     6    3    1    2  |  101|10|0|1|
  0  |  6 4 1 2  |     6    4    1    2  |     6    4    1    2  |  101|11|0|1|
  0  |  6 4 2 2  |     6    4    2    2  |     6    4    2    2  |  101|11|1|1|
  1  |  3 3 1 1  |     6    6    2    2  |     6    5    2    2  |  10|10|0|0|
  1  |  3 3 1 2  |     6    6    2    4  |     6    5    2    3  |  10|10|0|1|
  1  |  4 3 1 2  |     8    6    2    4  |     7    5    2    3  |  11|10|0|1|
  1  |  5 3 1 2  |    10    6    2    4  |     9    5    2    3  |  100|10|0|1|
  1  |  6 3 1 2  |    12    6    2    4  |    11    5    2    3  |  101|10|0|1|
  1  |  6 3 2 2  |    12    6    4    4  |    13    6    3    3  |  101|10|1|1|
  1  |  7 3 2 2  |    14    6    4    4  |    14    6    3    3  |  110|10|1|1|
  1  |  7 4 2 2  |    14    8    4    4  |    14    7    3    3  |  110|11|1|1|
  1  |  8 4 2 2  |    16    8    4    4  |    15    7    3    3  |  111|11|1|1|
  2  |  4 2 2 1  |    16    8    8    4  |    16    7    4    3  |  11|1|1|0|
  2  |  4 3 2 1  |    16   12    8    4  |    19   11    4    4  |  11|10|1|0|
  2  |  4 4 2 1  |    16   16    8    4  |    19   12    4    4  |  11|11|1|0|
  2  |  4 5 2 1  |    16   20    8    4  |    23   13    6    5  |  11|100|1|0|
  2  |  4 6 2 1  |    16   24    8    4  |    25   14    7    5  |  11|101|1|0|
  3  |  3 3 1 0  |    24   24    8    0  |    29   14    7    5  |  10|10|0||
  3  |  4 3 1 0  |    32   24    8    0  |    30   14    7    5  |  11|10|0||
  3  |  5 3 1 0  |    40   24    8    0  |    31   14    7    5  |  100|10|0||
  3  |  5 4 1 0  |    40   32    8    0  |    36   16    7    6  |  100|11|0||
  3  |  5 4 2 0  |    40   32   16    0  |    45   20    9    7  |  100|11|1||
  3  |  6 4 2 0  |    48   32   16    0  |    50   25   11    8  |  101|11|1||
  3  |  7 4 2 0  |    56   32   16    0  |    52   28   12    8  |  110|11|1||
  3  |  7 4 3 0  |    56   32   24    0  |    52   28   13    8  |  110|11|10||
  3  |  8 4 3 0  |    64   32   24    0  |    57   33   14   12  |  111|11|10||
  4  |  4 2 1 0  |    64   32   16    0  |    64   34   15   14  |  11|1|0||
  4  |  5 2 1 0  |    80   32   16    0  |    76   43   17   17  |  100|1|0||
  4  |  6 2 1 0  |    96   32   16    0  |    84   46   19   18  |  101|1|0||
  4  |  7 2 1 0  |   112   32   16    0  |    90   49   20   18  |  110|1|0||
  4  |  7 3 1 0  |   112   48   16    0  |    90   50   20   18  |  110|10|0||
  4  |  8 3 1 0  |   128   48   16    0  |    93   51   22   19  |  111|10|0||
  4  |  9 3 1 0  |   144   48   16    0  |    94   53   22   19  |  1000|10|0||
  4  | 10 3 1 0  |   160   48   16    0  |   125   63   25   25  |  1001|10|0||
  4  | 10 4 1 0  |   160   64   16    0  |   134   69   28   29  |  1001|11|0||
  5  |  5 2 1 1  |   160   64   32   32  |   137   69   28   30  |  100|1|0|0|
  5  |  5 3 1 1  |   160   96   32   32  |   139   73   29   30  |  100|10|0|0|
  5  |  5 4 1 1  |   160  128   32   32  |   154   76   32   30  |  100|11|0|0|
  5  |  5 4 2 1  |   160  128   64   32  |   216  102   52   44  |  100|11|1|0|
  5  |  6 4 2 1  |   192  128   64   32  |   234  105   55   46  |  101|11|1|0|
  5  |  7 4 2 1  |   224  128   64   32  |   236  105   55   46  |  110|11|1|0|
  5  |  8 4 2 1  |   256  128   64   32  |   277  118   62   55  |  111|11|1|0|
  6  |  4 2 1 0  |   256  128   64    0  |   279  118   63   55  |  11|1|0||
  6  |  4 2 1 1  |   256  128   64   64  |   315  133   70   64  |  11|1|0|0|
  6  |  5 2 1 1  |   320  128   64   64  |   322  139   73   65  |  100|1|0|0|
  6  |  6 2 1 1  |   384  128   64   64  |   345  154   79   66  |  101|1|0|0|
  6  |  6 2 1 2  |   384  128   64  128  |   374  176   88   76  |  101|1|0|1|
  6  |  7 2 1 2  |   448  128   64  128  |   399  186   99   87  |  110|1|0|1|
  6  |  8 2 1 2  |   512  128   64  128  |   435  195  105   90  |  111|1|0|1|
  6  |  9 2 1 2  |   576  128   64  128  |   478  225  120  107  |  1000|1|0|1|
  7  |  4 1 0 1  |   512  128    0  128  |   593  300  134  136  |  11|0||0|
  7  |  5 1 0 1  |   640  128    0  128  |   599  302  139  138  |  100|0||0|
  7  |  5 1 1 1  |   640  128  128  128  |   634  321  146  146  |  100|0|0|0|
  7  |  6 1 1 1  |   768  128  128  128  |   686  351  153  156  |  101|0|0|0|
  7  |  6 2 1 1  |   768  256  128  128  |   687  353  155  157  |  101|1|0|0|
  7  |  6 3 1 1  |   768  384  128  128  |   732  377  170  174  |  101|10|0|0|
  7  |  7 3 1 1  |   896  384  128  128  |   785  419  183  189  |  110|10|0|0|
  7  |  7 4 1 1  |   896  512  128  128  |   869  464  203  216  |  110|11|0|0|
  7  |  8 4 1 1  |  1024  512  128  128  |   969  522  235  236  |  111|11|0|0|
  7  |  8 4 1 2  |  1024  512  128  256  |   994  536  244  240  |  111|11|0|1|
  8  |  4 2 0 1  |  1024  512    0  256  |  1058  570  259  252  |  11|1||0|
  8  |  4 3 0 1  |  1024  768    0  256  |  1102  596  269  262  |  11|10||0|
  8  |  5 3 0 1  |  1280  768    0  256  |  1138  623  282  274  |  100|10||0|
  8  |  6 3 0 1  |  1536  768    0  256  |  1217  660  304  295  |  101|10||0|
  8  |  7 3 0 1  |  1792  768    0  256  |  1222  665  307  296  |  110|10||0|
  8  |  7 4 0 1  |  1792 1024    0  256  |  1344  718  339  321  |  110|11||0|
  8  |  7 5 0 1  |  1792 1280    0  256  |  1521  812  377  366  |  110|100||0|
  8  |  7 5 0 2  |  1792 1280    0  512  |  1581  850  393  384  |  110|100||1|
  9  |  3 3 0 1  |  1536 1536    0  512  |  1713  907  427  423  |  10|10||0|
  9  |  4 3 0 1  |  2048 1536    0  512  |  1743  922  430  427  |  11|10||0|
  9  |  5 3 0 1  |  2560 1536    0  512  |  1755  930  434  434  |  100|10||0|
  9  |  5 4 0 1  |  2560 2048    0  512  |  1849  987  460  453  |  100|11||0|
  9  |  6 4 0 1  |  3072 2048    0  512  |  2215 1195  572  548  |  101|11||0|
  9  |  6 4 1 1  |  3072 2048  512  512  |  2451 1293  630  607  |  101|11|0|0|
  9  |  7 4 1 1  |  3584 2048  512  512  |  2481 1311  640  615  |  110|11|0|0|
  9  |  8 4 1 1  |  4096 2048  512  512  |  2663 1396  691  660  |  111|11|0|0|
  9  |  8 4 2 1  |  4096 2048 1024  512  |  2751 1433  719  686  |  111|11|1|0|
 10  |  4 2 1 1  |  4096 2048 1024 1024  |  2817 1458  733  695  |  11|1|0|0|
 10  |  4 2 2 1  |  4096 2048 2048 1024  |  3437 1781  876  836  |  11|1|1|0|
 10  |  4 2 2 2  |  4096 2048 2048 2048  |  4039 2071 1030  979  |  11|1|1|1|
 10  |  4 3 2 2  |  4096 3072 2048 2048  |  4217 2145 1071 1021  |  11|10|1|1|
\end{verbatim}   
\end{document}